\documentclass[letterpaper, 10 pt]{ieeeconf}  
 
\IEEEoverridecommandlockouts 
\usepackage{geometry}
\newgeometry{top=0.75in, bottom=0.75in, right=0.75in, left=0.75in}

\newcommand{\bigzono}[1]{\Big\langle #1 \Big\rangle}

\newcommand{\zono}[1]{\langle #1 \rangle}
\usepackage{packagearticleieee}
\usepackage{amssymb}

\newtheorem{proposition}{Proposition}

\usepackage{cite}
\usepackage{booktabs}

 \usepackage{algorithm} 
 \usepackage{algpseudocode}
 \usepackage[utf8]{inputenc}
\usepackage{graphicx}
\usepackage{subcaption}

\usepackage{enumitem}

\makeatletter
\let\NAT@parse\undefined
\makeatother
\usepackage{url}
\usepackage[colorlinks=true,linkcolor=blue!80!black,citecolor=blue!80!black,urlcolor=blue!80!black]{hyperref}
\newcommand{\Rn}{\R^n}
\newcommand{\operator}[1]{{\normalfont \texttt{#1}}}

\newcommand{\CE}{R}
\newcommand{\Comp}{\textit{Complexity: }}

\DeclareMathSymbol{\shortminus}{\mathbin}{AMSa}{"39}

\newcommand{\id}{\mathsf{id}}

\def\tr{\text{tr}}

\title{\LARGE\bf Data-Driven Nonconvex Reachability Analysis \\ Using Exact Multiplication}

\author{Zhen Zhang$^1$, M. Umar B. Niazi$^2$, Michelle S. Chong$^3$, Karl H. Johansson$^2$, and Amr Alanwar$^1$
\thanks{$^1$ School of Computation, Information and Technology, Technical University of Munich, Germany. (Email: $\{$zhenzhang.zhang, alanwar$\}$@tum.de)}
\thanks{$^2$ Division of Decision and Control Systems, Digital Futures, School of Electrical Engineering and Computer Science, KTH Royal Institute of Technology, Stockholm, Sweden  (Email: $\{$mubniazi, kallej$\}$@kth.se)}
\thanks{$^3$ Department of Mechanical Engineering, Eindhoven University of Technology, the Netherlands (Email: m.s.t.chong@tue.nl)}
\thanks{M. U. B. Niazi is supported by the European Union’s Horizon Research and Innovation Programme under Marie Skłodowska-Curie grant agreement No. 101062523.}
}

\begin{document}
\maketitle
\thispagestyle{empty}
\pagestyle{empty}

\begin{abstract}

This paper addresses a fundamental challenge in data-driven reachability analysis: accurately representing and propagating non-convex reachable sets. We propose a novel approach using constrained polynomial zonotopes to describe reachable sets for unknown LTI systems. Unlike constrained zonotopes commonly used in existing literature, constrained polynomial zonotopes are closed under multiplication with constrained matrix zonotopes. We leverage this property to develop an exact multiplication method that preserves the non-convex geometry of reachable sets without resorting to approximations. We demonstrate that our approach provides tighter over-approximations of reachable sets for LTI systems compared to conventional methods.
\end{abstract}

\section{Introduction}

Reachability analysis provides the set of all possible states a system can transition to from a given initial state set, considering uncertainties within defined bounds or according to specified probabilistic models. 
It plays an essential role in providing formal guarantees required for safety-critical applications such as autonomous vehicles, robotics, and aerospace systems. 
Traditional model-based approaches for reachability require accurate mathematical models, which are often difficult to obtain for complex systems. 
Data-driven reachability analysis offers a promising alternative by leveraging measured data instead of relying on models. 

Data-driven approaches for reachability can be categorized into deterministic frameworks with bounded uncertainties and stochastic frameworks with probabilistic uncertainty models. 
Within the deterministic framework, several set representations have been employed to obtain overapproximated reachable sets. 
For instance, \cite{alanwar2021data} employs matrix zonotopes for representing the set of possible models consistent with noisy data, which is then used to approximate the reachable set as a constrained zonotope, a set representation introduced by \cite{scott2016constrained}. 
To incorporate prior knowledge and reduce conservatism, \cite{Alanwar2023Datadriven} employs constrained matrix zonotopes and extends the method to polynomial systems and Lipschitz nonlinear systems. 
Similarly, \cite{park2024data} uses linear time-varying approximations with bounded error terms for nonlinear systems and proposes an ellipsoidal over-approximation of the reachable set by solving a convex optimization problem. For hybrid or logical systems, \cite{siefert2022robust} and \cite{alanwar2025polynomial} introduce hybrid zonotopes and logical polynomial zonotopes, respectively, for representing reachable sets.


Within the stochastic framework, \cite{thorpe2019model, thorpe2021sreachtools} use kernel distribution embeddings in reproducing kernel Hilbert spaces for data-driven reachability, while \cite{devonport2023data} apply empirical inverse Christoffel functions with probably approximately correct (PAC) guarantees, and \cite{griffioen2023data} employ Gaussian process state space models for exact probability measures of trajectories.
For human-in-the-loop systems, \cite{choi2023data} and \cite{govindarajan2017data} use Gaussian mixture models and optimized disturbance bounds, respectively.
Recent advances also leverage deep learning and statistical methods, e.g., \cite{sivaramakrishnan2024stochastic} approximate stochastic reachability via neural networks, and \cite{hashemi2024statistical} address distribution shift in stochastic cyber-physical systems.

Deterministic approaches to data-driven reachability are preferred when constraint violations are unacceptable, e.g., collision avoidance. 
This is why regulatory frameworks for safety-critical systems (e.g., DO-178C \cite{rierson2017developing}: Software Considerations in Airborne Systems and Equipment Certification) often require deterministic guarantees rather than probabilistic ones. 
Furthermore, when data is limited or uncertainty distributions are unknown, deterministic approaches avoid potentially incorrect distributional assumptions, leading to robust guarantees even against worst-case scenarios.


Handling non-convex sets in deterministic reachability is challenging, e.g., obstacle-induced initial sets.
While convex sets have a computational advantage, they sacrifice accuracy by overapproximating the reachable set \cite{Alanwar2023Datadriven}. 
Recent works \cite{dietrich2024nonconvex, wang2023data, djeumou2022fly} have addressed non-convex reachable sets to some degree.
However, accurately representing and propagating non-convex sets without introducing excessive conservatism or computational complexity remains an open problem.
 
The challenge of non-convexity in data-driven reachability is particularly evident in linear systems where one obtains a reachable set by multiplying a set of system matrices consistent with data with a set of possible states. 
For instance, in \cite{Alanwar2023Datadriven}, reachable sets and a set of possible system matrices consistent with data are described using constrained zonotopes (CZs) and constrained matrix zonotopes (CMZs), respectively. 
However, CZs are not closed under multiplication with CMZs. This is the reason why convex set-based approaches to reachability typically resort to approximations that overbound the resulting set, leading to conservatism.

In this paper, we address this challenge in deterministic data-driven reachability by employing constrained polynomial zonotopes (CPZs) recently introduced by \cite{kochdumper2023constrained}. 
While \cite{luo2023reachability} presents an exact multiplication for matrix zonotopes (MZs) and polynomial zonotopes (PZs), it does not incorporate constraints into the multiplication itself. 
Therefore, our key contribution is to add the constraints of the sets in the exact multiplication operation by leveraging the closure property of CPZs when multiplied with CMZs. 
This extension preserves the non-convex geometry without introducing conservatism, thereby enabling more accurate safety verification for systems that can be more precisely characterized as CMZs---as opposed to just MZs---when dealing with non-convex reachable sets.


To summarize, our contributions are twofold. First, we provide a refinement method that updates the set of possible models as new data becomes available, enabling more accurate iterative reachability analysis. For this, we provide an exact method for intersecting two CMZs. Second, we introduce constrained polynomial matrix zonotopes (CPMZs) as a novel representation and provide an exact multiplication operation between a CPMZ and a CPZ. This operation is crucial for propagating non-convex sets of states accurately without introducing the conservatism typically associated with approximation techniques. This addresses the fundamental challenge in data-driven reachability analysis of maintaining the precise geometry of non-convex reachable sets.


\section{Preliminaries and Problem Formulation}\label{sec:preliminaries}

\subsection{Notations}
The sets of real and natural numbers are denoted by $\R$ and $\N$, respectively, with $\N_0 = \N\cup \{0\}$. 
We denote the matrix of zeros and ones of size $m \times n$ by $0_{m \times n}$ and $1_{m \times n}$, respectively, and the identity matrix by $I_n\in\R^{n\times n}$, where the subscripts are sometimes omitted to avoid clutter.
For a matrix $A$, $A^\top$ denotes the transpose and $A^\dagger$ the Moore–Penrose pseudoinverse.
By $A_{(i,j)}$, we denote the $(i,j)$-th entry and by $A_{(\cdot,j)}$, the $j$-th column. 
A slight abuse of notation is adopted for indexed matrices (e.g., $A_n$), where we use $A_n^{(i,j)}$ and $A_n^{(\cdot,j)}$.
For a vector or list $v$, $v_{(i)}$ denotes the $i$-th element, and $v_{(p_1:p_2)} \triangleq (v_{(p_1)}, \dots, v_{(p_2)})$ its restriction. 
The vectorization of a matrix $A \in \R^{n \times m}$ is denoted by $\mathrm{vec}(A) \in \R^{nm \times 1}$.
For a set $\mathcal{H} = \{1,2,\dots,{|\mathcal{H}|}\}$ with $|\mathcal{H}|$ denoting its cardinality, $A_{(\cdot,[1:|\mathcal{H}|])}$ denotes the matrix $[A_{(\cdot,1)},\dots, A_{(\cdot,|\mathcal{H}|)}]$.
The notation $[~]$ denotes an empty matrix or vector. The operators $\boxplus$ and $\otimes$ denote exact addition and exact multiplication, respectively.
\vspace{-1em}

\begin{table}[b]
\vspace{-1em}
    \centering
    \begin{tabular}{ll}
        \toprule
        CZ & Constrained Zonotope \\
        CPZ & Constrained Polynomial Zonotope \\
        MZ & Matrix Zonotope \\
        CMZ & Constrained Matrix Zonotope \\
        CPMZ & Constrained Polynomial Matrix Zonotope \\
        \bottomrule
    \end{tabular}
    \caption{List of acronyms}
    \label{tab:acronyms}
\end{table}

\subsection{Problem Statement}

We consider a linear time-invariant system in discrete-time $k\in\N_0$ with unknown system matrices $\Phi_\tr$ and $\Gamma_\tr$:
\vspace{-0.5em}
\begin{equation} \label{eq:model-linear}
    x_{(k)} = \Phi_\tr x_{(k-1)} + \Gamma_\tr u_{(k-1)} + w_{(k)},
\end{equation}
where $x_{(k)} \in \R^{n_x}$ is the state at time $k$, $u_{(k)} \in \R^{n_u}$ is the control input, and $w_{(k)} \in \R^{n_x}$ is the unknown noise.
Reachability analysis computes the set containing all possible states $x_{(k)}$ that can be reached from a compact set $\mathcal{X}_0$ containing initial states given a compact set $\mathcal{U}_k\ni u_{(k)}$ of feasible inputs and a compact set $\mathcal{Z}_w\ni w_{(k)}$ of possible noise.

\begin{definition}
The \textit{exact reachable set}
\begin{multline} 
    \label{eq:R}
    \mathcal{R}_{N} \triangleq \big\{ x_{(N)} \in \R^{n_x} \mid x_{(k+1)} = \Phi_\tr x_{(k)} + \Gamma_\tr u_{(k)} + w_{(k)}, \\
    x_{(0)} \in \mathcal{X}_0,
    u_{(k)} \in \mathcal{U}_k, w_{(k)} \in \mathcal{Z}_w, k = 0,\dots,N-1 \big\}
\end{multline}
is the set of all states that can be reached after $N$ time steps starting from initial set $\mathcal{X}_0$
subject to inputs $u_{(k)} \in \mathcal{U}_k$ and noise $w_{(k)} \in \mathcal{Z}_w$, for $k=0, \dots, N-1$.
\hfill $\lrcorner$
\end{definition}

Our goal is to estimate the exact reachable set of the system \eqref{eq:model-linear} using input-state data, without explicit knowledge of the system matrices $\Phi_\tr$ and $\Gamma_\tr$. 
Moreover, we are interested in the case where the initial set $\mathcal{X}_{0}$ is non-convex, and we must compute a reachable set that preserves this non-convex geometry without resorting to convex relaxations.

\subsection{Set Representations}
For our data-driven reachability analysis, we define the required set representations below. 

\begin{definition}
\label{def:conZonotope}
Given an offset $c\!\in\!\Rn$, generator $G\!\in\!\R^{n \times h}$, and constraints $A \in \R^{n_c \times h}$ and $b \in \R^{n_c}$, a \textit{constrained zonotope} (CZ) is a compact set in $\Rn$ defined as (see \cite{scott2016constrained})
\begin{equation}
    \label{eq:con-zon}
    \mathcal{C}{=}\bigg\{\!c\!+\! \sum_{k=1}^{h}\! \alpha_{(k)} G_{(\cdot,k)}\! \biggm| \! \sum_{k=1}^h \!\alpha_{(k)} A_{(\cdot,k)}\!=\!b,  \alpha_{(k)}\!\in\![\shortminus 1,\!1]\! \bigg \}.
\end{equation}
Furthermore, for reasons that will become apparent later, we associate an identifier $\id \in  \N^{1 \times h }$ with $\mathcal{C}$ for identifying the factors $\alpha_{(1)},\dots,\alpha_{(h)}$ in \eqref{eq:con-zon}.
To describe a CZ, we use a shorthand notation $\mathcal{C} = \zono{c,G,A,b,\id}_\text{CZ}$. 
\hfill $\lrcorner$
\end{definition}

Note that zonotopes are a special case of CZs, where constraints are empty \cite{kochdumper2023constrained}, i.e., $A=[~]$ and $B=[~]$.

\begin{definition}
Given an offset $c\!\in\!\Rn$, generator $G\!\in\!\R^{n \times h}$, exponent $E \in \N_0^{p \times h}$, and constraints $A \in \R^{n_c \times q}$, $b \in \R^{n_c}$, and $\CE \in \N_0^{p \times q}$, a \textit{constrained polynomial zonotope} (CPZ) is a compact set in $\Rn$ defined as (see \cite{kochdumper2023constrained})
\begin{multline}
    \label{eq:con-poly-zono}
    \mathcal{P} = \bigg \{ c + \sum_{i=1}^{h} \bigg( \prod_{k=1}^p \alpha_{(k)}^{E_{(k,i)}} \bigg) G_{(\cdot,i)} ~ \bigg | \\ 
    \sum_{i=1}^{q} \bigg( \prod_{k=1}^p \alpha_{(k)}^{\CE_{(k,i)}} \bigg) A_{(\cdot,i)} = b, \alpha_{(k)} \in [\shortminus 1,1]   \bigg \}.
\end{multline}

Similar to CZs, we associate an identifier $\id \in  \N^{1 \times p }$ with $\mathcal{P}$ for identifying the factors $\alpha_{(1)},\dots,\alpha_{(p)}$ in \eqref{eq:con-poly-zono}. 
We denote a CPZ as $\mathcal{P} = \zono{c,G,E,A,b,R,\id}_\text{CPZ}$.
\hfill $\lrcorner$
\end{definition}

\begin{example}
To understand the role of $\id$, consider
\begin{equation}
\label{eq:barp1}
    \mathcal{P}_1\!=\!\Biggl\langle\!    
    \begingroup
    \setlength\arraycolsep{1pt}
    \begin{bmatrix} 0 \\ 2 \\ 1 \end{bmatrix}\!,\!
    \begin{bmatrix} 0 & 1 \\ 3 & 2 \\ 1 & 5 \end{bmatrix}\!,\!
    \begin{bmatrix} 4 & 1 \\ 0 & 2 \end{bmatrix}\!,\!
    \begin{bmatrix} 1 & 2 \\ 0 & 0 \\ 3 & 4 \end{bmatrix}\!,\!
    \begin{bmatrix} 2 \\ 0 \\ 2 \end{bmatrix}\!,\!
    \begin{bmatrix} 4 & 2 \\ 0 & 2 \end{bmatrix}\!,\!
    \begin{bmatrix} 1 & 2 \end{bmatrix}\!
    \endgroup
    \Biggr\rangle_\text{CPZ}\!
\end{equation}
which describes the following CPZ
\vspace{-0.5em}
\begin{multline}
    \mathcal{P}_1 = \Bigg\{ 
    \begin{bmatrix} 0 \\ 2 \\ 1 \end{bmatrix} 
    +\begin{bmatrix} 0 \\ 3 \\ 1 \end{bmatrix} 
    \alpha_{(1)}^4 
    +\begin{bmatrix} 1 \\ 2 \\ 5 \end{bmatrix} 
    \alpha_{(1)} \alpha_{(2)}^2  \Biggm| \\ 
    \begin{bmatrix} 1 \\ 0 \\ 3 \end{bmatrix}
    \alpha_{(1)}^4
    \!+\!\begin{bmatrix} 2 \\ 0 \\ 4 \end{bmatrix}
    \alpha_{(1)}^2\alpha_{(2)}^2
    \!=\!\begin{bmatrix} 2 \\ 0 \\ 2 \end{bmatrix},
    \alpha_{(1)},\alpha_{(2)} \!\in\! [\shortminus 1, 1]\Bigg\}.
\end{multline}
Here, $\id = \begin{bmatrix} 1 & 2\end{bmatrix}$ identifies the dependent factor $\alpha_{(1)}$ with $\id_{(1)}=1$ and $\alpha_{(2)}$ with $\id_{(2)}=2$.
\hfill $\lrcorner$
\end{example}

\begin{definition} \label{def:conmatzonotopes}  
Given an offset $C \in \R^{m \times n}$, generators $G=\begin{bmatrix} G_{(1)}\!\dots\! G_{(\gamma)}\end{bmatrix} \in \R^{m \times (n\gamma)}$, constraints $A=\begin{bmatrix} A_{(1)} \!\dots\! A_{(\gamma)}\end{bmatrix} \in \R^{n_c \times (n_a\gamma)}$, and $B \in \R^{n_c \times n_a}$, a \textit{constrained matrix zonotope} (CMZ) is a compact set in $\R^{m\times n}$ defined as (see \cite{Alanwar2023Datadriven})
\begin{equation}
    \label{eq:con-mat-zono}
    \!\mathcal{N}\! =\!\Big\{\! C\! + \!\sum_{k=1}^{\gamma} \alpha_{(k)} G_{(k)}\! \Bigm|\!
    \sum_{k=1}^{\gamma} \!\alpha_{(k)} A_{(k)}\! =\! \!B,  \!\alpha_{(k)} \!\in\! [\shortminus 1,\!1] \! \Big\} .
\end{equation}
We associate an identifier $\id \in  \N^{1 \times p }$ with $\mathcal{N}$ for identifying the factors $\!\alpha_{(1)},\!\dots\!,\!\alpha_{(p)}$.
Denote $\mathcal{N} \!= \!\zono{C,G,A,B,\id}_\text{CMZ}$. 
\hfill $\lrcorner$
\end{definition}

Note that matrix zonotopes are a special case of CMZs, denoted by $\mathcal{N}=\zono{C,G,[~],[~],\id}_\text{CMZ}$. 

\begin{definition}
\label{def:conmatpolyzonotopes}  
Given an offset $C\in \R^{m \times n}$, generators $G=\begin{bmatrix} G_{(1)} \dots G_{(\gamma)} \end{bmatrix} \!\in\! \R^{m \times (n  \gamma)}$, exponent $E \!\in\! \N_0^{p \times \gamma}$, constraints $A\!=\!\begin{bmatrix} A_{(1)} \dots A_{(\gamma)} \end{bmatrix} \!\in\! \R^{n_c \times (n_a \gamma)}$, $B \!\in\! \R^{n_c \times n_a}$, and $\CE \!\in\! \N_0^{p \times \gamma}$, a \textit{constrained polynomial matrix zonotope} (CPMZ) is 
\vspace{-1.5em}
\begin{multline}
    \label{eq:con-poly-mat-zono}
    \mathcal{Y} = \Big\{ C + \sum_{i=1}^{\gamma} \bigg( \prod_{k=1}^p \alpha_{(k)}^{E_{(k,i)}} \bigg) \, G_{(i)} \Bigm| \\
    \sum_{i=1}^{\gamma} \bigg( \prod_{k=1}^p \alpha_{(k)}^{\CE_{(k,i)}} \bigg) A_{(i)} = B \,,   \alpha_{(k)} \in [\shortminus 1,1] \Big\} \; .
\end{multline}
We associate an identifier $\id \in \N^{1\times p}$ with $\mathcal{Y}$ for identifying the factors $\alpha_{(1)},\dots,\alpha_{(p)}$.
Furthermore, we use the shorthand notation $\mathcal{Y} = \zono{C,G,E,A,B,R,\id}_\text{CPMZ}$. 
\hfill $\lrcorner$
\end{definition}

To perform operations between two CPZs, the \operator{mergeID} operator is needed to make them compatible with each other.

\begin{proposition}[\operator{mergeID} \cite{kochdumper2020sparse}] \label{prop:mergeID}  
Given two CPZs
\vspace{-0.25em}
\begin{align*}
\mathcal{P}_1 &= \zono{c_1, G_1, E_1, A_1, b_1, R_1,  \id_1}_\text{CPZ} \\
\mathcal{P}_2 &= \zono{c_2, G_2, E_2, A_2, b_2, R_2, \id_2}_\text{CPZ}
\end{align*}
the \operator{mergeID} operator returns two adjusted CPZs that are equivalent to $\mathcal{P}_1$ and $\mathcal{P}_2$:
\vspace{-0.25em}
\begin{multline*}
    \operator{mergeID}(\mathcal{P}_1,\mathcal{P}_2) = \big \{ \langle c_1, G_1, \overline{E}_1, A_1, b_1, \overline{R}_1, \overline{\id} \rangle_\text{CPZ}, \\
    \langle c_2, G_2, \overline{E}_2,A_2,b_2,\overline{R}_2, \overline{\id} \rangle_\text{CPZ} \big \}
\end{multline*}
with $\overline{\id} = \begin{bmatrix} \id_1 & \id_2^{(\cdot,\mathcal{H})} \end{bmatrix}$, $\mathcal{H} = \left\{ i~ |~ \id_2^{(i)} \not\in \id_1 \right\}$, and
\begin{align*}
    \overline{E}_1 &= \begin{bmatrix} E_1 \\ 0_{|\mathcal{H}|\times h_1} \end{bmatrix} \in \R^{a \times h_1}, \qquad
    \overline{R}_1 = \begin{bmatrix} R_1 \\ 0_{|\mathcal{H}|\times q_1} \end{bmatrix} \in \R^{a \times q_1} \\
    \overline{E}_{2}^{(i,\cdot)} &= \begin{cases} E_2^{(j,\cdot)}, & \mathrm{if} ~ \exists j~\overline{\id}_{(i)} = \id_2^{(j)} \\ 
    0_{1\times h_2}, & \mathrm{otherwise} 
    \end{cases}\\
    \overline{R}_2^{(i,\cdot)} &= \begin{cases} R_2^{(j,\cdot)}, & \mathrm{if} ~ \exists j~\overline{\id}_{(i)} = \id_2^{(j)} \\ 
    0_{1\times q_2}, & \mathrm{otherwise}
    \end{cases}
\end{align*}
where $i = 1, \dots, a$ with $a = |\mathcal{H}|+p_1$ for $\id_1 \in  \N^{1 \times p_1 }$.
\hfill $\lrcorner$
\end{proposition}

\begin{example}
    Consider a CPZ
\begin{displaymath}
    {\mathcal{P}}_2 \!=\! \Biggl \langle\!
    \begingroup
    \setlength\arraycolsep{2pt}
    \begin{bmatrix}
    3\\3\\4
    \end{bmatrix}\!,\!\begin{bmatrix}
    2& 2\\3 & 0\\1 & 4
    \end{bmatrix}\!,\!\begin{bmatrix}
    3 & 2\\3 & 0
    \end{bmatrix}\!,\!\begin{bmatrix}
    1& 3\\2 & 4
    \end{bmatrix}\!,\!\begin{bmatrix}
    2\\ 5
    \end{bmatrix}\!,\! 
    \begin{bmatrix}
    2 & 0\\2 & 3
    \end{bmatrix}\!,\!\begin{bmatrix}
    1  & 3
    \end{bmatrix}
    \endgroup
    \!\Biggr \rangle_\text{CPZ}
\end{displaymath}
describing the following set
\begin{multline}
 {\mathcal{P}}_2 = \Bigg\{ \begin{bmatrix}
 3\\3\\4
\end{bmatrix} + \begin{bmatrix}
 2\\3\\1
\end{bmatrix} \alpha_{(1)}^3\alpha_{(2)}^3 + \begin{bmatrix}
 2\\0\\4
\end{bmatrix} \alpha_{(1)}^2  \Biggm| \\
\alpha_{(1)}^2\alpha_{(2)}^2\begin{bmatrix}
 1\\2
\end{bmatrix} 
+ \alpha_{(2)}^3\begin{bmatrix}
 3\\4
\end{bmatrix}=\begin{bmatrix}
 2\\5
\end{bmatrix}, \alpha_{(1)},\alpha_{(2)} \in[-1,1]\Bigg\}
\end{multline}
where $\id = \begin{bmatrix} 1 & 3\end{bmatrix}$ identifies the dependent factor $\alpha_{(1)}$ with $\id_{(1)}=1$ and $\alpha_{(2)}$ with $\id_{(2)}=3$. If we apply the operator \operator{mergeID}(${\mathcal{P}}_1, {\mathcal{P}}_2$) where ${\mathcal{P}}_1$ is defined in \eqref{eq:barp1}, we get the following sets with common identifiers.
\begin{align*}
    \bar{\mathcal{P}}_1 &= \Biggl \langle\!
    \begingroup
    \setlength\arraycolsep{2pt}
    \begin{bmatrix}
    0\\2\\1
    \end{bmatrix}\!,\!\begin{bmatrix}
    0 & 1\\3 & 2\\1 & 5
    \end{bmatrix}\!,\!\begin{bmatrix}
    4 & 1\\0 & 2\\0 & 0 
    \end{bmatrix}\!,\!\begin{bmatrix}
    1 & 2\\0 & 0\\3 & 4
    \end{bmatrix}\!,\!\begin{bmatrix}
    2\\0\\2 
    \end{bmatrix},\begin{bmatrix}
    4 & 2\\0 & 2\\0 & 0 
    \end{bmatrix}\!,\!\begin{bmatrix}
    1 &2 & 3
    \end{bmatrix} 
    \endgroup
    \!\Biggr \rangle \nonumber
\\
    \bar{\mathcal{P}}_2 &= \Biggl \langle\!
    \begingroup
    \setlength\arraycolsep{2pt}
    \begin{bmatrix}
    3\\3\\4
    \end{bmatrix}\!,\!\begin{bmatrix}
    2& 2\\3 & 0\\1 & 4
    \end{bmatrix}\!,\!\begin{bmatrix}
    3 & 2\\0 & 0\\ 3& 0
    \end{bmatrix}\!,\!\begin{bmatrix}
    1 & 3\\2 & 4
    \end{bmatrix}\!,\!\begin{bmatrix}
    2\\5
    \end{bmatrix}\!,\!\begin{bmatrix}
    2 & 0\\0 & 0\\2 & 3 
    \end{bmatrix}\!,\!\begin{bmatrix}
    1 &2 &3
    \end{bmatrix}  
    \endgroup
    \!\Biggr \rangle.
\end{align*}
\hfill $\lrcorner$
\end{example}

To preserve dependencies between generators during the addition of two CPZs and prevent over-approximation, we adopt the exact addition operation proposed by \cite{kochdumper2020sparse}.
By retaining dependency information, the resulting sets are less conservative, improving the computational efficiency and reliability of the resulting CPZ.

\begin{definition}
Let $\mathcal{P}_1 = \zono{c_1,G_1,E_1,A_1,b_1,R_1,\id_1}_\text{CPZ} \subset \R^{n}$ and $\mathcal{P}_2 = \zono{c_2,G_2,E_2,A_2,b_2,R_2,\id_2}_\text{CPZ} \subset \R^{n}$. Then, the exact addition is given by 
\begin{multline}
\mathcal{P}_1 \boxplus \mathcal{P}_2 = \bigzono{
\begingroup
\setlength\arraycolsep{2pt}
\left[\begin{array}{l}
c_1 \\
c_2
\end{array}\right],\left[\begin{array}{cc}
G_1 & 0_{n \times h_2} \\
0_{n \times h_1} & G_2
\end{array}\right],\left[\begin{array}{cc}
\overline{E}_1 &  \overline{E}_2
\end{array}\right], 
\endgroup \\
\begingroup
\setlength\arraycolsep{2pt}
\left[\begin{array}{cc}
A_1 & 0_{n_{c_1} \times q_2} \\
0_{n_{c_2} \times q_1} & A_2
\end{array}\right],\left[\begin{array}{l}
b_1 \\
b_2
\end{array}\right],\left[\begin{array}{cc}
\overline{R}_1 & \overline{R}_2
\end{array}\right],\id_{12}
\endgroup
}_\text{CPZ}
\end{multline}
where $\big \{ \overline{E}_1,\overline{E}_2,\overline{R}_1,\overline{R}_2,\id_{12} \big \} \gets \operator{mergeID}(\mathcal{P}_1,\mathcal{P}_2)$.
	\label{prop:exactAddition}
\hfill $\lrcorner$
\end{definition}
\subsection{Assumptions}
We assume access to $n_T \in \N$ input-state trajectories of length $T_i + 1$, with inputs $\{ u_{(k)}\}^{T_i-1}_{k=0}$ and states $\{ x_{(k)}\}^{T_i}_{k=0}$ for $i\in\{0,\dots,n_T-1\}$. For simplicity, we consider the offline data from a single trajectory ($n_T = 1$) of length $T_0$ and organize the input and
noisy state data into matrices:
\begin{equation}
    \label{eq:offline-data}
    \begin{array}{c}
    X_0^+ = \begin{bmatrix}  x_{(1)} & x_{(2)} & \dots & x_{(T_0)}\end{bmatrix},\\[2pt]
    X_0^- = \begin{bmatrix}  x_{(0)}& x_{(1)} & \dots & x_{(T_0-1)}\end{bmatrix}, \\[2pt]
    U_0^- = \begin{bmatrix} u_{(0)} &  u_{(1)}  & \dots & u_{(T_0-1)} \end{bmatrix}.
    \end{array}  
\end{equation} 
Let $D_0^-\!=\! 
\begingroup
\setlength\arraycolsep{2pt}
\begin{bmatrix} X_0^{-\top} & U_0^{-\top} \end{bmatrix}^\top
\endgroup$ and $D_0\!=\! 
\begingroup
\setlength\arraycolsep{2pt}
\begin{bmatrix} X_0^{+\top} & X_0^{-\top} &U_0^{-\top}\end{bmatrix}^\top
\endgroup$. 

\begin{assumption}\label{ass:zon-noise}
    For all $k\in \Z_{\geq 0}$, the noise $w_{(k)}$ is bounded by a known zonotope $\mathcal{Z}_w$, which includes the origin.
    \hfill $\lrcorner$
\end{assumption}
\begin{assumption}\label{ass:rank_D}
    The data matrix $D_0^-$ has full row rank, i.e., $\mathrm{rank}(D_0^-) = n_x + n_u$.
    \hfill $\lrcorner$
\end{assumption}

We represent the sequence of unknown noise as $\{w_{(k)}\}_{k=0}^{T_0}$. From Assumption~\ref{ass:zon-noise}, it follows that 
\[
W_0^- = \begin{bmatrix} w_{(0)} & \cdots & w_{(T_0-1)} \end{bmatrix} \in \mzon_w = \langle C_{\mzon_w}, G_{\mzon_w} \rangle,
\]
where $C_{\mzon_w} \in \R^{n_x \times n_T}$ and $G_{\mzon_w} \in \R^{n_x \times \gamma_{\mathcal{Z}_w} n_T}$ with $\mzon_w$ denoting the MZ resulting from the concatenation of multiple noise zonotopes $\mathcal{Z}_w$~\cite{Alanwar2023Datadriven}. 
Finally, let all system matrices $\begin{bmatrix} \Phi & \Gamma \end{bmatrix}$ of \eqref{eq:model-linear} consistent with the data $D_-$ be
\begin{equation}
    \mathcal{N}_{\Sigma} \!=\! \Big\{\! 
    \begingroup
    \setlength\arraycolsep{3pt}
    \begin{bmatrix} \Phi & \Gamma \end{bmatrix}
    \endgroup
    \!\Bigm|\! X_0^+ \!=\! \Phi X_0^- \!+\! \Gamma U_0^- \!+\! W_0^-, W_0^- \!\in\! \mathcal{M}_w \!\Big\}. \label{eq:Nsig}
\end{equation}
By definition, $\begin{bmatrix} \Phi_\tr & \Gamma_\tr \end{bmatrix} \in \mathcal{N}_{\Sigma}$.

\section{Data-driven Reachability Analysis with Non-convex Set Representations}\label{sec:DDSV}

Given system \eqref{eq:model-linear} with unknown matrices $\Phi_\tr$ and $\Gamma_\tr$, the presence of noise implies multiple feasible matrices $\begin{bmatrix} \Phi & \Gamma \end{bmatrix}$ are consistent with the observed data. 
For the reachability analysis to be reliable, one must incorporate all such data-consistent models.
To this end, we extend the offline approach proposed by \cite{Alanwar2023Datadriven} to continue refining the estimated set of models by leveraging online data.

Our proposed 
data-driven reachability analysis
algorithm 
comprises two steps: 
first, determining a set of models consistent with observed data given bounded uncertainties; and second, using this model set to compute the region in state space containing all possible true states.
Using offline data \eqref{eq:offline-data}, we compute an initial set of models represented by a CMZ. Then, in the online phase, we refine this CMZ iteratively by collecting online input-state trajectories.
For computing the reachable set as a CPZ in the online phase, we propose an exact multiplication method for the time update to multiply the CMZ representing a set of models with the CPZ representing the previous reachable set.

\begin{figure}
    \centering
    \includegraphics[width=1\linewidth]{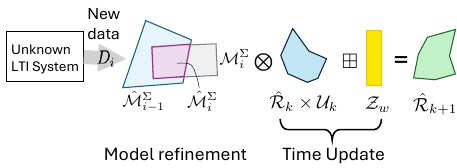}
    \caption{Data-driven reachability analysis with model refinement. Legend: $\otimes$ denotes the exact multiplication operation and $\boxplus$ denotes the exact addition operation.}
    \label{fig:algo-cartoon}
\end{figure}
\vspace{-10pt}
\subsection{Initial Set of Models}

Using the offline data \eqref{eq:offline-data}, we compute the initial set of models $\mathcal{M}_0^{\Sigma}$ that is consistent with the data as follows. 

\begin{lemma}[\cite{Alanwar2023Datadriven}]
\label{lm:sigmaM1}
Consider the offline input-state data \eqref{eq:offline-data} such that Assumption~\ref{ass:rank_D} holds and consider the MZ
\begin{align}
    \mathcal{M}_0^\Sigma = (X_0^+ - \mathcal{M}_w) \begin{bmatrix} 
    X_0^- \\ U_0^- 
    \end{bmatrix}^{\dagger}.
    \label{eq:zonoAB}
\end{align} 
Then, $\begin{bmatrix} \Phi_\tr & \Gamma_\tr \end{bmatrix} \in \mathcal{M}_0^\Sigma$ and $\mathcal{N}_\Sigma \subseteq \mathcal{M}_0^\Sigma$. 
\hfill $\lrcorner$
\end{lemma}
\begin{remark}
Requiring the full row rank in Assumption~\ref{ass:rank_D}, i.e., $\begin{bmatrix} X^{-\top} & U^{-\top}\end{bmatrix}^{\top}=n_x+n_u$, implies that there exists a right-inverse of the matrix $\left[\begin{array}{ll}
X^{-\top} & U^{-\top}
\end{array}\right]^{\top}$.
\hfill $\lrcorner$
\end{remark}

\subsection{Intersection Between CMZs}

For the following result, we omit the identifier $\id$ of CMZs as it is irrelevant when intersecting two sets.

\begin{proposition}[Intersection] 
Consider two CMZs 
\begin{align*}
    \mathcal{N}_1 &= \zono{C_1,G_1,A_1,B_1}_\text{CMZ}\subset \R^{n_x \times n} \\ 
    \mathcal{N}_2 &= \zono{C_2,G_2,A_2,B_2}_\text{CMZ}\subset\R^{n_x \times n}.
\end{align*}
Then, their intersection is given by
\begin{equation}
\mathcal{N}_1 \cap \mathcal{N}_2=  \bigzono{ C_1,\bigg[ G_1,0_{ n_x\times (n\gamma_2)}\bigg],\hat{A},\hat{B} }_\text{CMZ} 
\label{generalintersection}
\end{equation}
where 
\begin{displaymath}
\begingroup\setlength\arraycolsep{2pt}
\hat{A}=\left[\begin{array}{cc}
\hat{A}_1 & 0_{n_{c_1}n_{a_1} \times \gamma_2}\\
0_{n_{c_2}n_{a_2} \times \gamma_1}&\hat{A}_2\\
\hat{G}_1&-\hat{G}_2
\end{array}\right], ~ \hat{B}=\left[\begin{array}{c}
\text{vec}\big({B}_1\big) \\
\text{vec}\big({B}_2\big) \\
\text{vec}\big({C}_2-{C}_1\big)
\end{array}\right]
\endgroup
\end{displaymath}
with $\hat{A}_i= [\text{vec}({A}_i^{(1)}),\dots,\text{vec}({A}_i^{({\gamma_i})})]$ and
$\hat{G}_i= [\text{vec}({G}_i^{(1)}),\dots,\text{vec}({G}_i^{({\gamma_i})})]$, for $i=1,2$.
\hfill $\lrcorner$
\label{Intersection}
\end{proposition}
Although the proof is non-trivial, it follows lines of arguments similar to those of \cite[Proposition 1]{scott2016constrained}. It is therefore omitted due to space restrictions.

\Comp
The computations of $\text{vec}(B_1)$, $\text{vec}(B_2)$, and $\text{vec}(C_2 - C_1)$ require 
$\mathcal{O}(n_{c_1}n_{a_1})$, $\mathcal{O}(n_{c_2}n_{a_2})$, and $\mathcal{O}(2n_x n)$ operations, respectively.
The constructions of $\hat{A}_1$, $\hat{A}_2$, $\hat{G}_1$, and $\hat{G}_2$ involve complexities 
$\mathcal{O}(n_{c_1}n_{a_1}\gamma_1)$, $\mathcal{O}(n_{c_2}n_{a_2}\gamma_2)$, $\mathcal{O}(n_x n \gamma_1)$, and 
$\mathcal{O}(n_x n \gamma_2)$, respectively. Thus, the overall computational complexity of intersection \eqref{Intersection} is
$
\mathcal{O}\left(\gamma_1(n_{c_1}n_{a_1} + n_x n) + \gamma_2(n_{c_2}n_{a_2} + n_x n)\right).
$
\hfill $\diamond$

\begin{remark}
The intersection $\mathcal{N}_1 \cap \mathcal{N}_2 = \hat{\mathcal{N}}$ in \eqref{generalintersection} yields constraint matrices $\hat{A} \in \R^{\hat{m} \times (\gamma_1 + \gamma_2)}$ and $\hat{B} \in \R^{\hat{m} \times 1}$, with $\hat{m} = n_{c_1}n_{a_1} + n_{c_2}n_{a_2} + n_x n$. This corresponds to a special case of CMZs with $n_a = 1$, adopted to facilitate exact multiplications with CPZs in subsequent computations. When $\hat{m}/\hat{n}_c \in \N$, the vector $\hat{B}$ can be reshaped as a matrix $\hat{B} \in \R^{\hat{n}_c \times (\hat{m}/\hat{n}_c)}$ by the reshaping operator defined as follows:
\[
\operator{Convert}(\hat B\in \R^{\hat m \times 1},\hat n_c)=\hat{B}\in \R^{\hat n_{{c}} \times (\hat{m}/\hat{n}_c)}.
\]
Accordingly, the reshaped constraint matrix $\hat{A}$ lies in the space
$
\hat{A} \in \R^{\hat n_{{c}} \times \left( (\hat{m}/\hat{n}_c) \cdot (\gamma_1 + \gamma_2) \right)},
$
where
$
\hat{A}=\big[\operator{Convert}(\hat{A}^{(1)}), \dots,\operator{Convert}(\hat{A}^{(\gamma_1+\gamma_2)})\big].
$
\hfill $\lrcorner$
\end{remark}

\subsection{Iterative Refinement of Set of Models}

Similar to Lemma~\ref{lm:sigmaM1}, given any input-state data $D_i = (X_i^+, X_i^-, U_i^-)$ of system \eqref{eq:model-linear} over the interval $t \in [T_{i-1}, T_i]$, for $i\geq 1$, and that Assumption~\ref{ass:rank_D} holds, the resulting MZ $\mathcal{M}_i^\Sigma$ contains all matrices $\begin{bmatrix}\Phi & \Gamma\end{bmatrix}$ consistent with the data and noise bound, i.e., $\mathcal{N}_\Sigma \subseteq \mathcal{M}_i^\Sigma$.
At iteration~$i$ of the refinement phase ($i \geq 1$), we intersect the set of models 
${\mathcal{M}}_{i}^{\Sigma}$ derived from newly obtained input-state trajectories with the previous model set $\hat{\mathcal{M}}_{i-1}^{\Sigma}$ to get refined set $\hat{\mathcal{M}}_{i}^{\Sigma}$. The initial set of models and refined sets are subsequently used to compute an over-approximation of the reachable sets of the unknown system, as illustrated in Fig.~\ref{fig:algo-cartoon}.

\begin{proposition}
\label{prop:reach_lin}
Consider data $D_i = (X_i^+,X_i^-,U_i^-)$, for $i\geq 0$, with $D_0$ an initial data and $D_1,D_2,\dots$ subsequent new data obtained from system \eqref{eq:model-linear}, where each data $D_i$ satisfies Assumption~\ref{ass:rank_D} and results in a set of models $\mathcal{M}_i^\Sigma$ using Lemma~\ref{lm:sigmaM1}.
The refined set of models  
\begin{align}
&\hat{\mathcal{M}}_i^\Sigma=
\mathcal{M}_i^\Sigma \cap \hat{\mathcal{M}}_{i-1}^\Sigma, \quad i\geq 1,\quad \text{with}\quad \hat{\mathcal{M}}_0^\Sigma=
\mathcal{M}_0^\Sigma \label{cmzintersection}
\end{align}
is a CMZ and contains the true model $\begin{bmatrix} \Phi_\tr & \Gamma_\tr \end{bmatrix}$.
\end{proposition}

\begin{proof}
A straightforward consequence of Lemma~\ref{lm:sigmaM1}.
\end{proof}

\subsection{Exact Multiplication for Non-Convex Reachability}
At time $k$, as depicted in Fig.~\ref{fig:algo-cartoon}, we need to iteratively multiply $\hat{\mathcal{M}}_i^\Sigma$ with the cartesian product of reachable set $\hat{\mathcal{R}}_k$ and the input set $\mathcal{U}_k$ to obtain the next reachable set $\hat{\mathcal{R}}_{k+1}$.
Although $\hat{\mathcal{M}}_i^\Sigma$ is a CMZ, notice that CMZs are a special case CPMZ. 
Therefore, to be more general, we provide an exact multiplication method for CPMZ with CPZ. 

Consider ${\mathcal{Y}} = \zono{C_{\mathcal{Y}}, G_{\mathcal{Y}}, E_{\mathcal{Y}}, A_{\mathcal{Y}}, B_{\mathcal{Y}}, R_{\mathcal{Y}}, \id_{\mathcal{Y}}}_\text{CPMZ}\subset \R^{n_x \times n }$ and ${\mathcal{P}} = \zono{c_{\mathcal{P}}, G_{\mathcal{P}}, E_{\mathcal{P}}, A_{\mathcal{P}}, b_{\mathcal{P}}, R_{\mathcal{P}}, \id_{\mathcal{P}}}_\text{CPZ}\subset \R^{n}$. Since both sets may contain shared dependent factors, consistent indexing of exponent matrices is required. To address this, we leverage the \operator{mergeID} operator, which constructs a unified exponent representation to preserve these dependencies:
\begin{multline}
  \hspace{-7pt}
  \operator{mergeID}(\mathcal{Y},\mathcal{P}) \!=\! \big \{\! \underbrace{\langle C_{\mathcal{Y}},\! G_{\mathcal{Y}},\! \overline{E}_\mathcal{Y},\! A_{\mathcal{Y}},\! B_{\mathcal{Y}},\! \overline{R}_\mathcal{Y},\! \id_\mathcal{YP} \rangle_\text{CPMZ}}_{\bar{\mathcal{Y}}}, \\
  \underbrace{\langle c_{\mathcal{P}},\! G_{\mathcal{P}},\! \overline{E}_{\mathcal{P}},\! A_{\mathcal{P}},\! b_{\mathcal{P}},\! \overline{R}_{\mathcal{P}},\! \id_{\mathcal{YP}} \rangle_\text{CPZ}}_{\bar{\mathcal{P}}} \big \}.  \label{mergetwoexact}
\end{multline}

\begin{proposition}[Exact Multiplication] \label{prop:multi}  
Given a CPMZ ${\mathcal{Y}} = \zono{C_{\mathcal{Y}}, G_{\mathcal{Y}}, E_{\mathcal{Y}}, A_{\mathcal{Y}}, B_{\mathcal{Y}}, R_{\mathcal{Y}}, \id_{\mathcal{Y}}}_\text{CPMZ}\subset \R^{n_x \times n }$ and a CPZ $\mathcal{P} = \langle c_{\mathcal{P}}, G_{\mathcal{P}}, {E}_{\mathcal{P}}, A_{\mathcal{P}}, b_{\mathcal{P}}, {R}_{\mathcal{P}}, \id_{\mathcal{P}} \rangle_\text{CPZ}\subset \R^{n}$, the following identity holds
\begin{multline}
\mathcal{Y} \otimes \mathcal{P} =  \bigzono{C_{\mathcal{Y}} c_{\mathcal{P}},\begin{bmatrix} G_{\mathcal{Y}}c_{\mathcal{P}} & C_{\mathcal{Y}} G_{\mathcal{P}}& G_{f} \end{bmatrix} , E_{\mathcal{YP}}, \\
    A_{\mathcal{Y}\mathcal{P}}, {B}_{\mathcal{Y}\mathcal{P}} , R_{\mathcal{YP}} ,\id_{\mathcal{YP}}}_\text{CPZ},  \label{eq:matczono}
\end{multline}
where $\mathcal{Y}\otimes \mathcal{P} \subset \R^{n_x}$ and
\begin{align*}
    A_{\mathcal{Y}\mathcal{P}} &= \begin{bmatrix} \text{vec}(A^{(1)}_{\mathcal{Y}}) & \dots & \text{vec}
    (A_{\mathcal{Y}}^{(\gamma)}) & 0_{n_{c}n_{a}  \times q_{\mathcal{P}}} \\
     0_{m_{\mathcal{P}}  \times 1} & \dots & 0_{m_{\mathcal{P}}  \times 1} & A_{\mathcal{P}}
    \end{bmatrix} \\
    {B}_{\mathcal{Y}\mathcal{P}} &=\begin{bmatrix} \text{vec}(B_{\mathcal{Y}}) \\  b_{\mathcal{P}}\end{bmatrix} \\
    E_\mathcal{YP} &= \bigg[ \overline{E}_{\mathcal{Y}}, \overline{E}_{\mathcal{P}}, 
    \Big[\overline{E}_{\mathcal{Y}}^{(\cdot,1)} + \overline{E}_{\mathcal{P}}^{(\cdot,1)}\Big],
    \dots,
    \Big[\overline{E}_{\mathcal{Y}}^{(\cdot,1)} + \overline{E}_{\mathcal{P}}^{(\cdot,h_\mathcal{P})}\Big], \\
    & \dots, \Big[\overline{E}_{\mathcal{Y}}^{(\cdot,\gamma)} + \overline{E}_{\mathcal{P}}^{(\cdot,1)}\Big],
    \dots, \Big[\overline{E}_{\mathcal{Y}}^{(\cdot,\gamma)}+\overline{E}_{\mathcal{P}}^{(\cdot,h_\mathcal{P})}\Big] \bigg]  
    \\
    {R_\mathcal{YP}} &= \begin{bmatrix} \overline{R}_{\mathcal{Y}} & \overline{R}_{\mathcal{P}} \end{bmatrix} \\
    G_{f} &= \begin{bmatrix} g_{f}^{(1)} & \dots & g_{f}^{(h_\mathcal{P}\gamma)} \end{bmatrix}
\end{align*}
with $\overline{E}_\mathcal{Y},\overline{E}_\mathcal{P},\overline{R}_\mathcal{Y},\overline{R}_\mathcal{P},\id_{\mathcal{YP}}$ obtained from \eqref{mergetwoexact} and
\begin{displaymath}
    g_{f}^{(k)} =  G^{(i)}_{\mathcal{Y}} G_{\mathcal{P}}^{(\cdot,j)}, \quad k=h_{\mathcal{P}} (i-1) + j,
\end{displaymath}
for $i=1,\dots,\gamma$,\;$j=1,\dots,h_\mathcal{P}$.
\end{proposition}
\begin{proof}
    Let $\hat{\mathcal{P}}$ be the right-hand side of \eqref{eq:matczono} and let $Y\in \mathcal{Y}$ and $p \in \mathcal{P}$. 
    We will prove that $\mathcal{Y}\mathcal{P} \subseteq \hat{\mathcal{P}}$ and $\hat{\mathcal{P}} \subseteq \mathcal{Y}\mathcal{P}$, for all $Y\in\mathcal{Y}$ and $p\in\mathcal{P}$. 
    Note that with the implementation of $\operator{mergeID}(\mathcal{Y},\mathcal{P})$, $Y$ and $p$ can be written as:
    \begin{align}
    \exists \hat{\alpha}_{([1:a])} &: & Y &= C_{\mathcal{Y}} + \sum_{i=1}^{\gamma}\bigg( \prod_{k=1}^{a}
    \hat{\alpha}_{(k)}^{\overline{E}_\mathcal{Y}^{(k,i)}} \bigg)  G_{\mathcal{Y}}^{(i)}
    \nonumber\\
    \exists \hat{\alpha}_{([1:a])} &: & p 
    &= \! c_{\mathcal{P}} +   \sum_{i=1}^{h_\mathcal{P}} \bigg( \prod_{k=1}^{a}
    \hat{\alpha}_{(k)}^{\overline{E}_\mathcal{P}^{(k,i)}} \bigg) G_{\mathcal{P}}^{(\cdot,i)} \label{alphadetail}
    \end{align}
    where $a = |\mathcal{H}|+\gamma$, $\mathcal{H} = \left\{ i~ |~ \id_{{\mathcal{P}}(i)} \not\in \id_{\mathcal{Y}} \right\}$. 
    Thus, we have
    \begin{multline}
    Yp = C_{\mathcal{Y}} c_{\mathcal{P}} +  \sum_{i=1}^{\gamma}\bigg( \prod_{k=1}^{a}
    \hat{\alpha}_{(k)}^{\overline{E}_\mathcal{Y}^{(k,i)}} \bigg)  G_{\mathcal{Y}}^{(i)}  c_{\mathcal{P}} \\
    + C_{\mathcal{Y}}\sum_{i=1}^{h_{\mathcal{P}}} \bigg( \prod_{k=1}^{a}
    \hat{\alpha}_{(k)}^{\overline{E}_\mathcal{P}^{(k,i)}} \bigg)  G_{\mathcal{P}}^{(\cdot,i)} \\
    + \sum_{i=1}^{\gamma} \sum_{j=1}^{h_{\mathcal{P}}} \bigg( \prod_{k=1}^{a}
    \hat{\alpha}_{(k)}^{\overline{E}_\mathcal{Y}^{(k,i)}} \bigg) \bigg( \prod_{k=1}^{a}
    \hat{\alpha}_{(k)}^{\overline{E}_\mathcal{P}^{(k,j)}} \bigg)  G_{\mathcal{Y}}^{(i)}  G_{\mathcal{P}}^{(\cdot,j)} . \label{eq:P_details}
    \end{multline}
    For the second and third terms on the right-hand side of \eqref{eq:P_details}, we have two sets of factors, each containing $\gamma$ and $h_\mathcal{P}$ elements. Consequently, the first $\gamma + h_\mathcal{P}$ entries of $\hat{\alpha}$ and the columns of $E_\mathcal{YP}$ are defined accordingly as follows:
    \begin{align}
    \hat{\alpha}_{([1:\gamma+h_\mathcal{P}])} &=\bigg[\prod_{k=1}^{a}
    \hat{\alpha}_{(k)}^{\overline{E}_\mathcal{Y}^{(k,1)}}, \dots,\prod_{k=1}^{a}
    \hat{\alpha}_{(k)}^{\overline{E}_\mathcal{Y}^{(k,\gamma)}} \nonumber\\
    & \hspace{1cm} 
    \prod_{k=1}^{a} \hat{\alpha}_{(k)}^{\overline{E}_\mathcal{P}^{(k,1)}}, \dots, 
    \prod_{k=1}^{a} \hat{\alpha}_{(k)}^{\overline{E}_\mathcal{P}^{(k,h_\mathcal{P})}}\bigg] \\{E}_{\mathcal{YP}}^{([1:\gamma+h_\mathcal{P}])} &= \Big[ \overline{E}_{\mathcal{Y}},\overline{E}_{\mathcal{P}}\Big].
\end{align} 
Because of $\operator{mergeID}(\mathcal{Y},\mathcal{P})$ in \eqref{mergetwoexact}, the fourth term on the right-hand side of \eqref{eq:P_details} can be expressed as:
\begin{multline}
\sum_{i=1}^{\gamma} \sum_{j=1}^{h_{\mathcal{P}}} \bigg( \prod_{k=1}^{a}
  \hat{\alpha}_{(k)}^{\overline{E}_\mathcal{Y}^{(k,i)}}  
  \hat{\alpha}_{(k)}^{\overline{E}_\mathcal{P}^{(k,j)}} \bigg)  G_{\mathcal{Y}}^{(i)}  G_{\mathcal{P}}^{(\cdot,j)}= \\
  \sum_{i=1}^{\gamma} \sum_{j=1}^{h_{\mathcal{P}}} \bigg( \prod_{k=1}^{a}
  \hat{\alpha}_{(k)}^{\overline{E}_\mathcal{Y}^{(k,i)}+  
  \overline{E}_\mathcal{P}^{(k,j)}} \bigg)  G_{\mathcal{Y}}^{(i)}  G_{\mathcal{P}}^{(\cdot,j)}. \label{1:n+p}
\end{multline} 
Concatenating the factors in \eqref{1:n+p}, we have
\begin{multline*}
\hat{\alpha}_{([\gamma+h_\mathcal{P}+1:\gamma+h_\mathcal{P}+\gamma h_\mathcal{P}])} \\ 
=\bigg[\prod_{k=1}^{a}
  \hat{\alpha}_{(k)}^{\overline{E}_\mathcal{Y}^{(k,1)}+  
  \overline{E}_\mathcal{P}^{(k,1)}}, \dots, \prod_{k=1}^{a}
  \hat{\alpha}_{(k)}^{\overline{E}_\mathcal{Y}^{(k,1)}+  
  \overline{E}_\mathcal{P}^{(k,h_\mathcal{P})}}, \dots, \\
  \prod_{k=1}^{a}
  \hat{\alpha}_{(k)}^{\overline{E}_\mathcal{Y}^{(k,\gamma)}+  \overline{E}_\mathcal{P}^{(k,1)}}, \dots,\prod_{k=1}^{a}
  \hat{\alpha}_{(k)}^{\overline{E}_\mathcal{Y}^{(k,\gamma)}+  \overline{E}_\mathcal{P}^{(k,h_\mathcal{P})}}\bigg],
\end{multline*}
 which results in $E_\mathcal{YP}^{([\gamma+h_\mathcal{P}+1:\gamma+h_\mathcal{P}+\gamma h_\mathcal{P}])}$ and $G_{f}$ as follows:
 \begin{align*}
    &E_\mathcal{YP}^{([\gamma+h_\mathcal{P} + 1 : \gamma+h_\mathcal{P} + \gamma h_\mathcal{P}])}= \bigg[  
    \Big[\overline{E}_{\mathcal{Y}}^{(\cdot,1)} + \overline{E}_{\mathcal{P}}^{(\cdot,1)}\Big],
    \dots, 
    \Big[\overline{E}_{\mathcal{Y}}^{(\cdot,1)} \\
    & + \overline{E}_{\mathcal{P}}^{(\cdot,h_\mathcal{P})}\Big], 
    \dots, \Big[\overline{E}_{\mathcal{Y}}^{(\cdot,\gamma)} + \overline{E}_{\mathcal{P}}^{(\cdot,1)}\Big],\dots, \Big[\overline{E}_{\mathcal{Y}}^{(\cdot,\gamma)}+\overline{E}_{\mathcal{P}}^{(\cdot,h_\mathcal{P})}\Big] \bigg] \\
    &G_{f} = \bigg[ G^{(1)}_{\mathcal{Y}} G_{\mathcal{P}}^{(\cdot,1)},{...}, G^{(1)}_{\mathcal{Y}} G_{\mathcal{P}}^{(\cdot,h_\mathcal{P})},{...},G^{(\gamma)}_{\mathcal{Y}} G_{\mathcal{P}}^{(\cdot,h_\mathcal{P})} \bigg].
    \end{align*}
    Secondly, we find the constraints on $\hat{\alpha}_{([1:a])}$ in \eqref{alphadetail}. For  $Y\in \mathcal{Y}$ and  $p \in \mathcal{P}$, $Yp$ should satisfy the constraints simultaneously.
    From \eqref{mergetwoexact}, we have
    \begin{multline}
    \sum_{i=1}^{\gamma} \bigg( \prod_{k=1}^{\gamma}
(\alpha_{\mathcal{Y}}^{(k)})^{R_\mathcal{Y}^{(k,i)}} \bigg) \text{vec}(A_{\mathcal{Y}}^{(i)}) =\\ 
    \sum_{i=1}^{\gamma} \bigg( \prod_{k=1}^{a}
    \hat{\alpha}_{(k)}^{\overline{R}_\mathcal{Y}^{(k,i)}} \bigg) \text{vec}(A_{\mathcal{Y}}^{(i)})=
    \text{vec}(B_{\mathcal{Y}})
    \label{constraint1}
    \end{multline} 
    and
    \begin{multline}
    \sum_{i=1}^{q_\mathcal{P}} \bigg( \prod_{k=1}^{h_\mathcal{P}}
(\alpha_{\mathcal{P}}^{(k)})^{R_\mathcal{P}^{(k,i)}} \bigg)A_{\mathcal{P}}^{{(\cdot,i)}} = \\
    \sum_{i=1}^{q_\mathcal{P}} \bigg( \prod_{k=1}^{a}
    \hat{\alpha}_{(k)}^{\overline{R}_\mathcal{P}^{(k,i)}} \bigg)A_{\mathcal{P}}^{{(\cdot,i)}}=b_{\mathcal{P}}.
    \label{constraint2}
    \end{multline}
    Combining \eqref{constraint1} and \eqref{constraint2}, the following holds
    \begin{multline}
    \sum_{i=1}^{\gamma} \bigg( \prod_{k=1}^{a}
    \hat{\alpha}_{(k)}^{\overline{R}_\mathcal{Y}^{(k,i)}} \bigg)\begin{bmatrix} \text{vec}(A_{\mathcal{Y}}^{{(\cdot,i)}}) \\  0_{m_\mathcal{P}\times1}\end{bmatrix}\\
    +\sum_{i=1}^{q_\mathcal{P}} \bigg( \prod_{k=1}^{a}
    \hat{\alpha}_{(k)}^{\overline{R}_\mathcal{P}^{(k,i)}} \bigg)\begin{bmatrix} 0_{n_{c}n_{a}\times1} \\  A_{\mathcal{P}}^{{(\cdot,i)}}\end{bmatrix}=\begin{bmatrix} \text{vec}(B_{\mathcal{Y}}) \\  b_{\mathcal{P}}\end{bmatrix}
    \label{constraint3} 
    \end{multline}
    which results in $A_{\mathcal{Y}\mathcal{P}}$ and $B_{\mathcal{Y}\mathcal{P}}$. 
    Thus, $Yp \in \hat{\mathcal{P}}$ and therefore $\mathcal{Y} \mathcal{P} \subseteq \hat{\mathcal{P}}$. 
    Conversely, let $\hat{p} \in \hat{\mathcal{P}}$, then 
    \begin{align*}
    \exists \hat{\alpha}_{([1:{\gamma+h_\mathcal{P}+\gamma h_\mathcal{P}}])}\!\!&:  \hat p =\nonumber\hat c + \sum_{i=1}^{\gamma+h_\mathcal{P}+\gamma h_\mathcal{P}}\bigg( \prod_{k=1}^{a}
    \hat{\alpha}_{(k)}^{{E}_\mathcal{YP}^{(k,i)}} \bigg)  G_{f}^{(i)}.
    \end{align*}
    By partitioning
    \begin{multline*}
    \hat{\alpha}_{([1:{\gamma+h_\mathcal{P}+\gamma h_\mathcal{P}}])}=\bigg[\prod_{k=1}^{a}
    \hat{\alpha}_{(k)}^{\overline{E}_\mathcal{Y}^{(k,[1:\gamma])}},\prod_{k=1}^{a}
    \hat{\alpha}_{(k)}^{\overline{E}_\mathcal{P}^{(k,[1:h_\mathcal{P}])}},\\
    \prod_{k=1}^{a}
    \hat{\alpha}_{(k)}^{\overline{E}_\mathcal{Y}^{(k,1)}+  
    \overline{E}_\mathcal{P}^{(k,1)}}, \dots,\prod_{k=1}^{a}
    \hat{\alpha}_{(k)}^{\overline{E}_\mathcal{Y}^{(k,1)}+  
    \overline{E}_\mathcal{P}^{(k,h_\mathcal{P})}},\dots,\\
    \prod_{k=1}^{a}
    \hat{\alpha}_{(k)}^{\overline{E}_\mathcal{Y}^{(k,\gamma)}+  \overline{E}_\mathcal{P}^{(k,1)}}, \dots,\prod_{k=1}^{a}
    \hat{\alpha}_{(k)}^{\overline{E}_\mathcal{Y}^{(k,\gamma)}+  \overline{E}_\mathcal{P}^{(k,h_\mathcal{P})}}\bigg]  
    \label{eq:beta_c1_3}
    \end{multline*}
    it follows that there exist $Y\in \mathcal{Y}$ and $p \in \mathcal{P}$ such that $\hat{p} = Yp$.
    Meanwhile, since $\hat{p}\in\hat{\mathcal{P}}$, it holds that
    \begin{align}
    \sum_{i=1}^{q_\mathcal{P}+\gamma} \bigg( \prod_{k=1}^{a}
    \hat{\alpha}_{(k)}^{{R}_\mathcal{YP}^{(k,i)}} \bigg)A_{\mathcal{YP}}^{{(\cdot,i)}}=\begin{bmatrix} \text{vec}(B_{\mathcal{Y}}) \\  b_{\mathcal{P}}\end{bmatrix}
    \end{align}
    which satisfies the constraints in \eqref{constraint3}. Therefore, $\hat{p} \in \mathcal{Y} \mathcal{P}$ and thus $ \hat{\mathcal{P}} \subseteq \mathcal{Y}   \mathcal{P}$.
\end{proof}

As noted before, $\mathcal{N}=\zono{C_{\mathcal{N}}, G_{\mathcal{N}}, A_{\mathcal{N}}, B_{\mathcal{N}}, \id_{\mathcal{N}}}_\text{CMZ}$ is a special case of CPMZ and thus can be reformulated as $\mathcal{Y}=\zono{C_{\mathcal{N}}, G_{\mathcal{N}}, E_{\mathcal{N}}, A_{\mathcal{N}}, B_{\mathcal{N}}, R_{\mathcal{N}}, \id_{\mathcal{N}}}_\text{CPMZ}$ with $E_{\mathcal{N}} = R_{\mathcal{N}} = I_{\gamma_{\mathcal{N}}}$. 
Then, according to Proposition~\ref{prop:multi}, the CPMZ is multiplied exactly with a CPZ to yield a new CPZ, as illustrated by Fig.~\ref{fig:proof-exact-mult-cartoon}.

\begin{figure}[!h]
    \centering
    \includegraphics[width=0.8\linewidth]{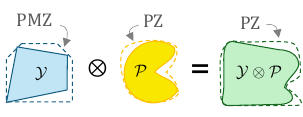}
    \caption{Illustration of the proof of Proposition \ref{prop:multi}, where the dashed lines indicate the boundary of PMZ and PZs, respectively. Let $\textrm{c}(\mathcal{Y})$ and $\textrm{c}(\mathcal{P})$ denote the constraints associated with CPMZ $\mathcal{Y}$ and CPZ $\mathcal{P}$, respectively. The exact multiplication of $\mathcal{Y}:=\{\textrm{PMZ}:\textrm{c}(\mathcal{Y}) \textrm{ hold}\}$ with $\mathcal{P}:=\{\textrm{PZ}:\textrm{c}(\mathcal{P}) \textrm{ hold}\}$ is a CPZ $\mathcal{Y}\otimes \mathcal{P}=\{\textrm{PZ}:\textrm{c}(\mathcal{Y}\otimes \mathcal{P}) \textrm{ hold}\}$, with constraints $\textrm{c}(\mathcal{Y}\otimes \mathcal{P}):= \textrm{Proj}_{\R^{n}} c(\mathcal{Y})$ and $c(\mathcal{P})$, where $\textrm{Proj}_{\R^{n}} c(\mathcal{Y})$ is the projection of $c(\mathcal{Y})$ from $\R^{n_x\times n}$ to $\R^{n}$.}
    \label{fig:proof-exact-mult-cartoon}
\end{figure}

\Comp Computation of $\text{vec}(B_\mathcal{Y})$, $A_{\mathcal{Y}\mathcal{P}}$, $E_{\mathcal{Y}\mathcal{P}}$, $C_\mathcal{Y} c_\mathcal{P}$, $G_\mathcal{Y} c_\mathcal{P}$, $C_\mathcal{Y} G_\mathcal{P}$, and $G_f$ has complexity $\mathcal{O}(n_c n_a)$, $\mathcal{O}(n_c n_a \gamma)$, $\mathcal{O}((|\mathcal{H}| + \gamma)\gamma h_\mathcal{P})$, $\mathcal{O}(n_x n)$, $\mathcal{O}(\gamma n_x n)$, $\mathcal{O}(n_x n h_\mathcal{P})$, and $\mathcal{O}(\gamma n_x n h_\mathcal{P})$, respectively. Then, we obtain the overall complexity of computing exact multiplication \eqref{eq:matczono} to be
$\mathcal{O}\left( 
  n_c n_a \gamma + 
  (|\mathcal{H}| + \gamma) \gamma h_\mathcal{P} +
  \gamma n_x n h_\mathcal{P}
\right)$.
\hfill $\diamond$

\begin{algorithm}[h]
\caption{Data-driven Non-Convex Reachability Analysis}
\label{alg:offline-online}
\textbf{Input:} 
Initial input-state data $D_0 = (X_0^+, X_0^-, U_0^-)$, initial set $\mathcal{X}_0$, noise zonotope $\mathcal{Z}_w$, input zonotopes $\mathcal{U}_k$; \\
\textbf{Output:} 
Reachable sets $\hat{\mathcal{R}}_k$ for $k \in \N$.

\vspace{0.5em}
\noindent\textbf{Offline: Initialization}
\begin{algorithmic}[1]
  \State $\hat{\mathcal{R}}_0 \gets \mathcal{X}_0$
  \State $\mathcal{M}_0^\Sigma \gets ( X_0^+ - \mathcal{M}_w ) \begin{bmatrix} X_0^- \\ U_0^- \end{bmatrix}^\dagger$
 \State Set $\hat{\mathcal{M}}^\Sigma \gets \mathcal{M}_0^\Sigma$
\State Set $X^+ \gets [~]$, $X^- \gets [~]$, $U^- \gets [~]$
\end{algorithmic}

\vspace{0.5em}
\noindent\textbf{Online: Set Refinement and Reachability Analysis }
\begin{algorithmic}[1]

\While{$k\geq 1$}
    \State Construct online data matrices $X^+ = [X^+ ~ x_{(k)}]$, $X^- = [X^- ~ x_{(k-1)}]$, $U^- =[U^- ~ u_{(k-1)}]$
    
    \If{$\textrm{rank}\left(\begin{bmatrix}X^- \\ U^-\end{bmatrix}\right)=n_x + n_u$}  
        \State $\mathcal{M}^\Sigma \gets ( X^+ - \mathcal{M}_w ) \begin{bmatrix} X^- \\ U^- \end{bmatrix}^\dagger$
    \State $\hat{\mathcal{M}}^\Sigma \gets \mathcal{M}^\Sigma \cap \hat{\mathcal{M}}^\Sigma$
    \State $X^+ \gets [~]$, $X^- \gets [~]$, $U^- \gets [~]$
    
\EndIf
\State $\hat{\mathcal{R}}_{k+1} \gets \hat{\mathcal{M}}^\Sigma\otimes ( \hat{\mathcal{R}}_k \times \mathcal{U}_k) \boxplus \mathcal{Z}_w$

\State $k \gets k+1$.
\EndWhile
  
\end{algorithmic}
\end{algorithm}

\begin{figure*}[t]
    \vspace{-1em}
    \centering
    \begin{subfigure}[h]{0.32\textwidth}
        \includegraphics[scale=0.26]{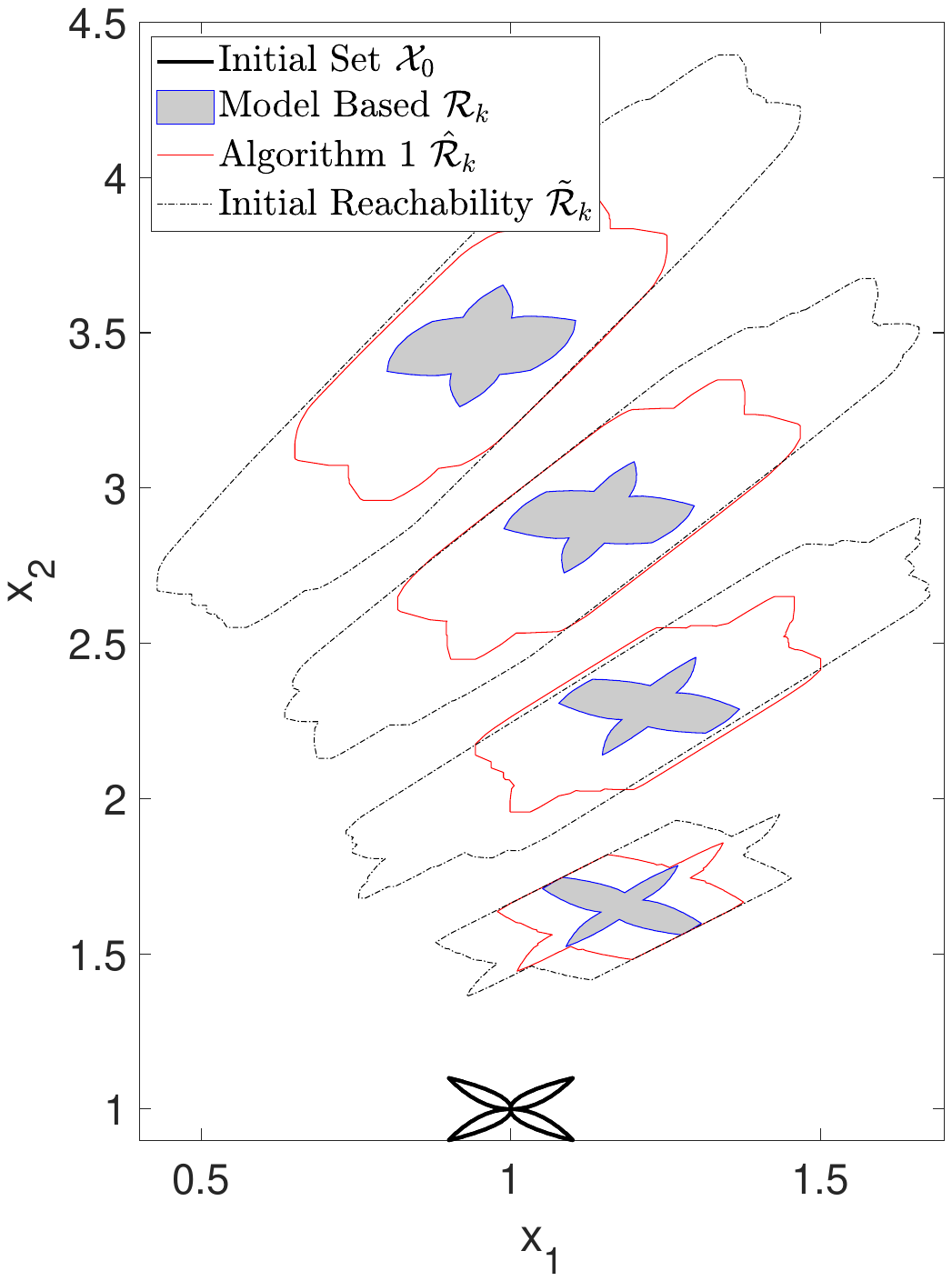}
        \caption{}
        \label{fig:x1x2_a}
    \end{subfigure}
    \begin{subfigure}[h]{0.32\textwidth}
        \includegraphics[scale=0.26]{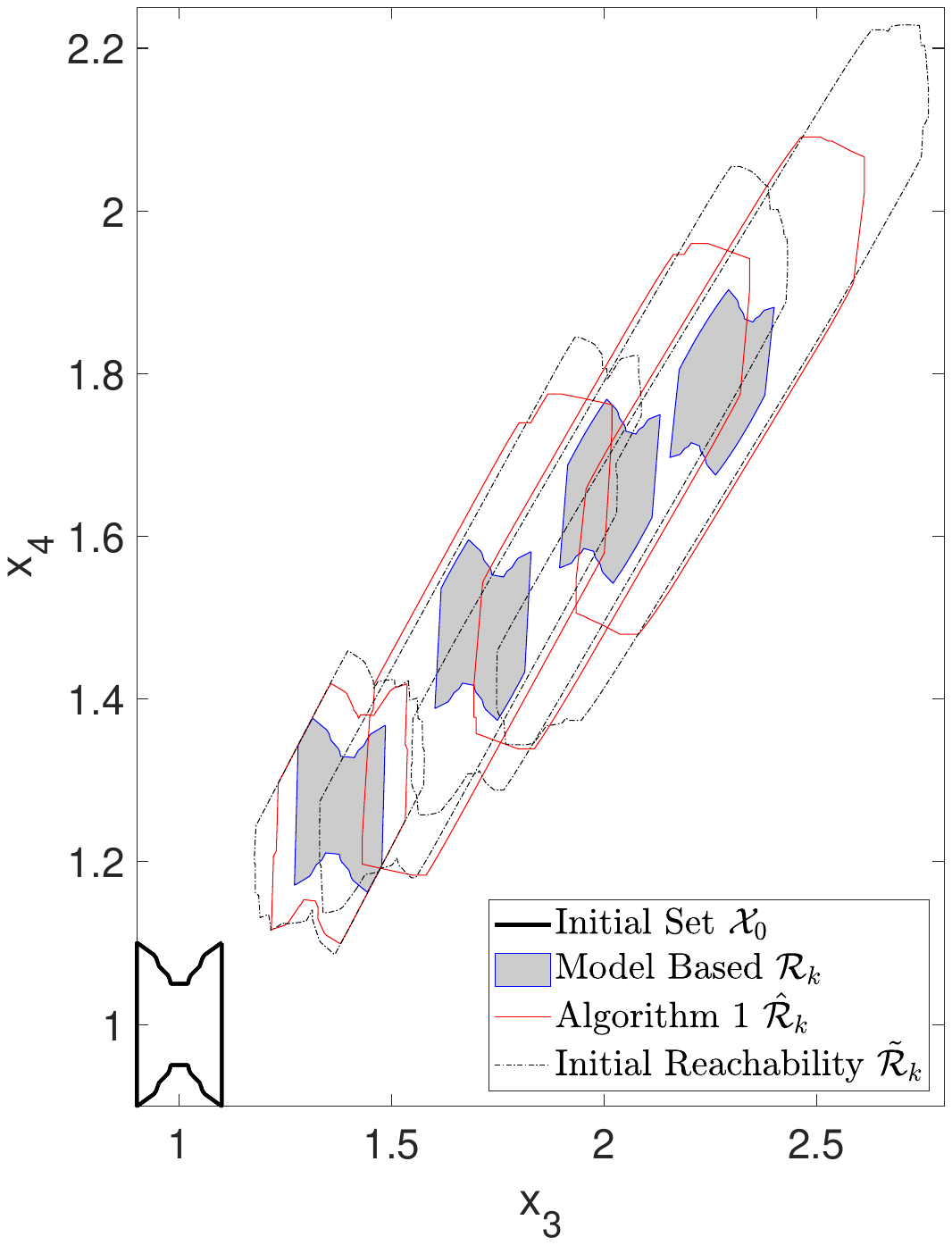}
        \caption{}
        \label{fig:x3x4_a}
    \end{subfigure}
    \begin{subfigure}[h]{0.32\textwidth}
        \includegraphics[scale=0.26]{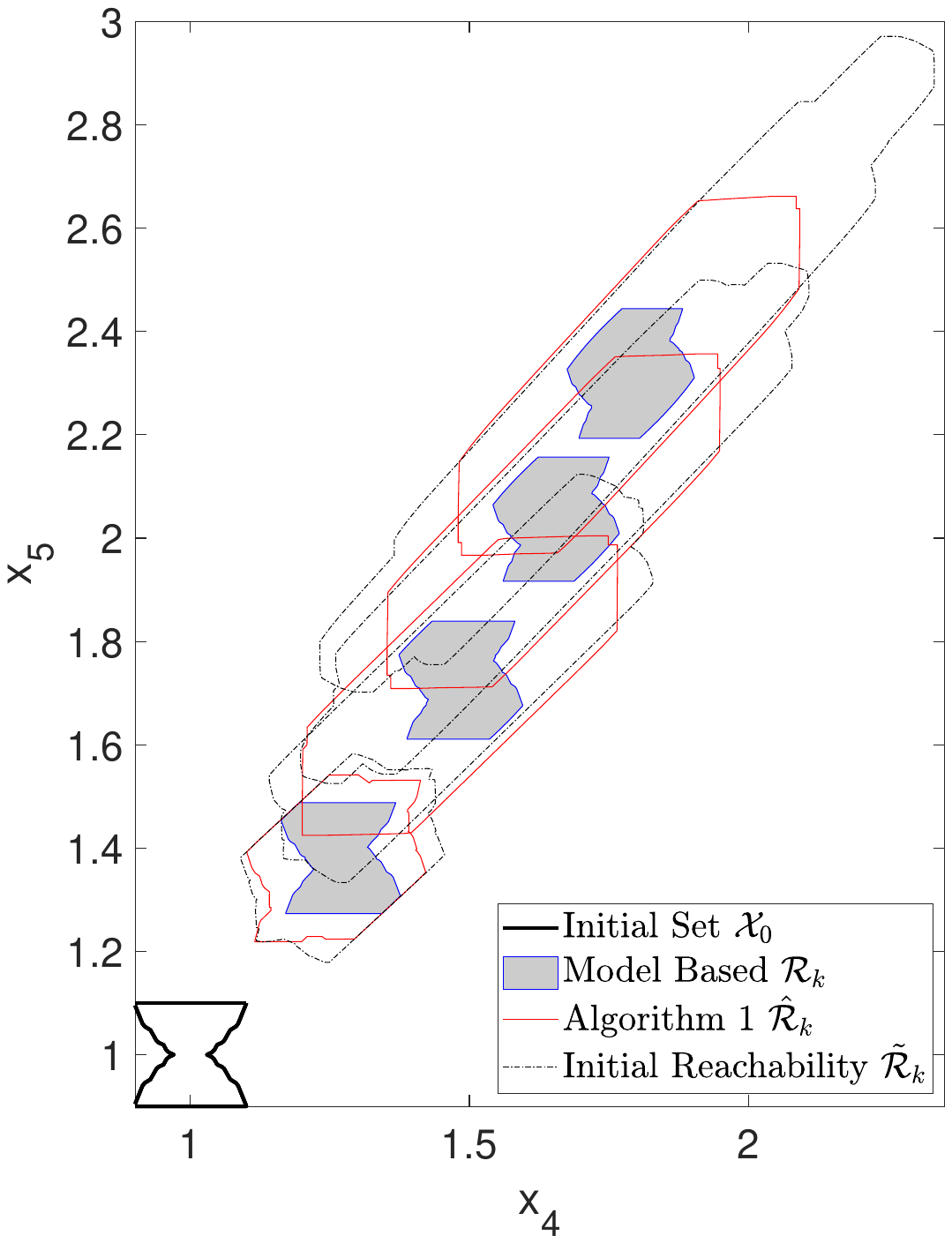}
        \caption{}
        \label{fig:x4x5_a}
    \end{subfigure}
    \caption{Projection of Reachable Sets Computed from Input-State Data Using Algorithm \ref{alg:offline-online} with a Non-Convex Initial Set}
    \label{fig:projSetA}
\end{figure*}
\begin{figure*}[!htbp]
    \centering
    \begin{subfigure}[h]{0.32\textwidth}
        \includegraphics[scale=0.26]{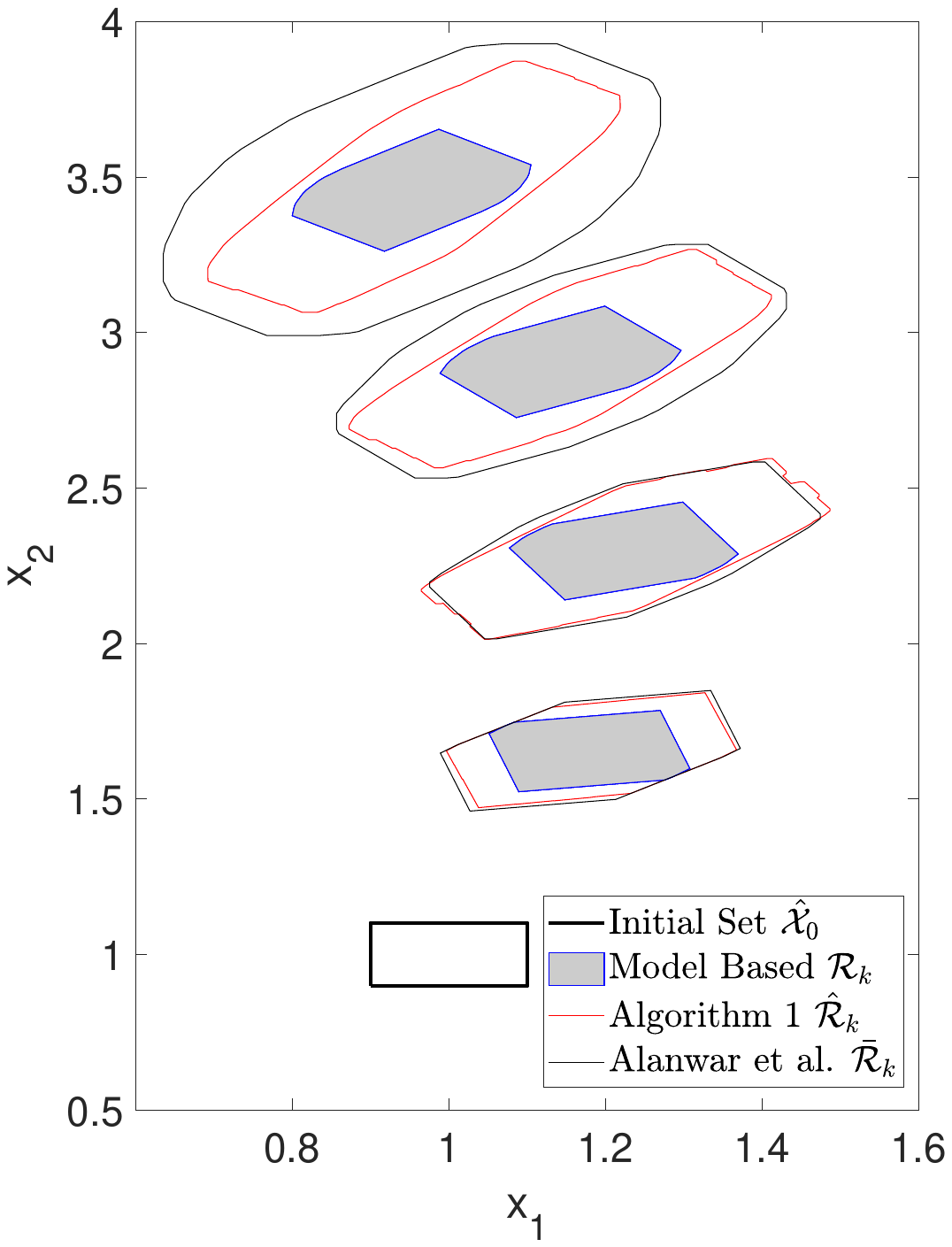}
        \caption{}
        \label{fig:x1x2_b}
    \end{subfigure}
    \begin{subfigure}[h]{0.32\textwidth}
        \includegraphics[scale=0.26]{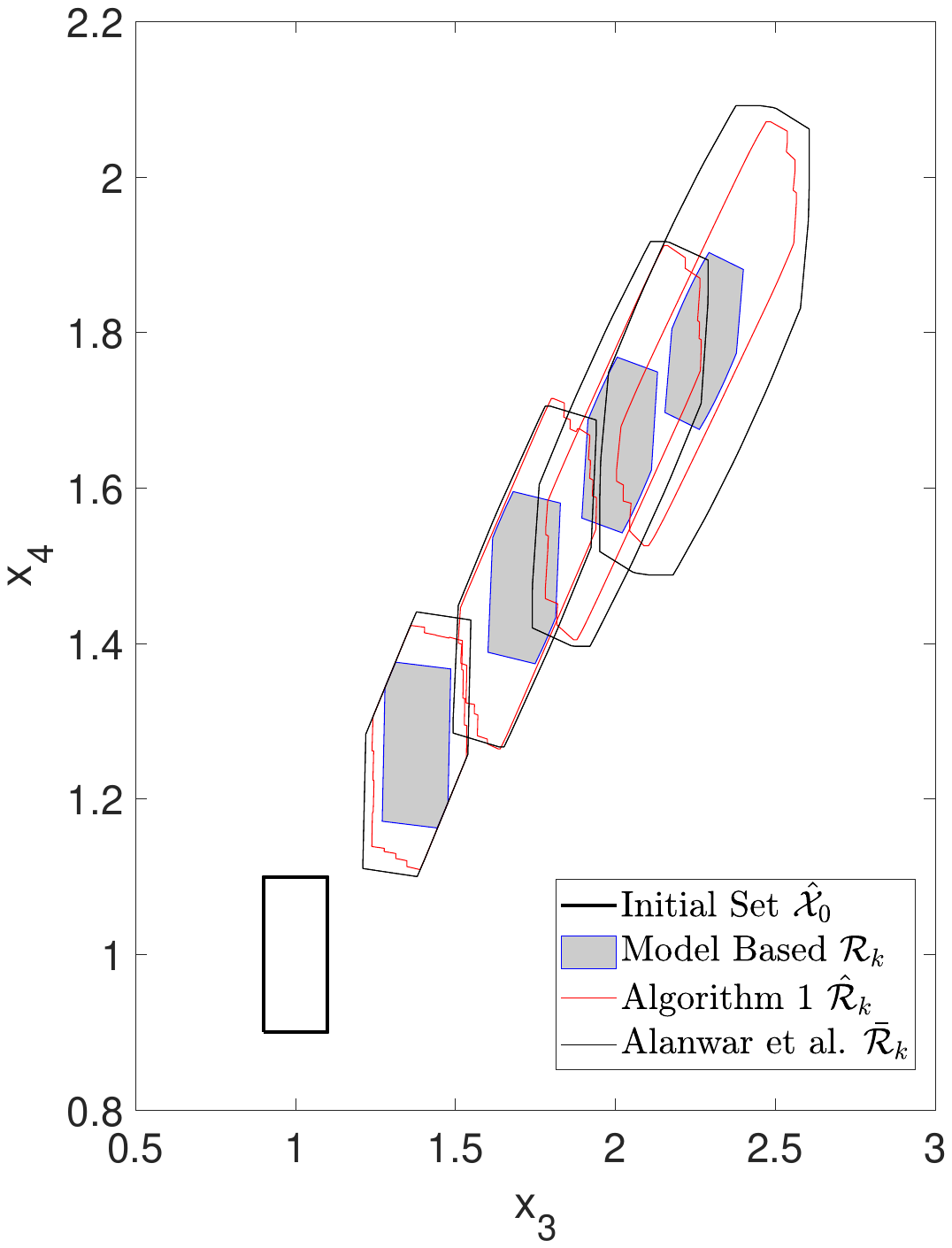}
        \caption{}
        \label{fig:x3x4_b}
    \end{subfigure}
    \begin{subfigure}[h]{0.32\textwidth}
        \includegraphics[scale=0.26]{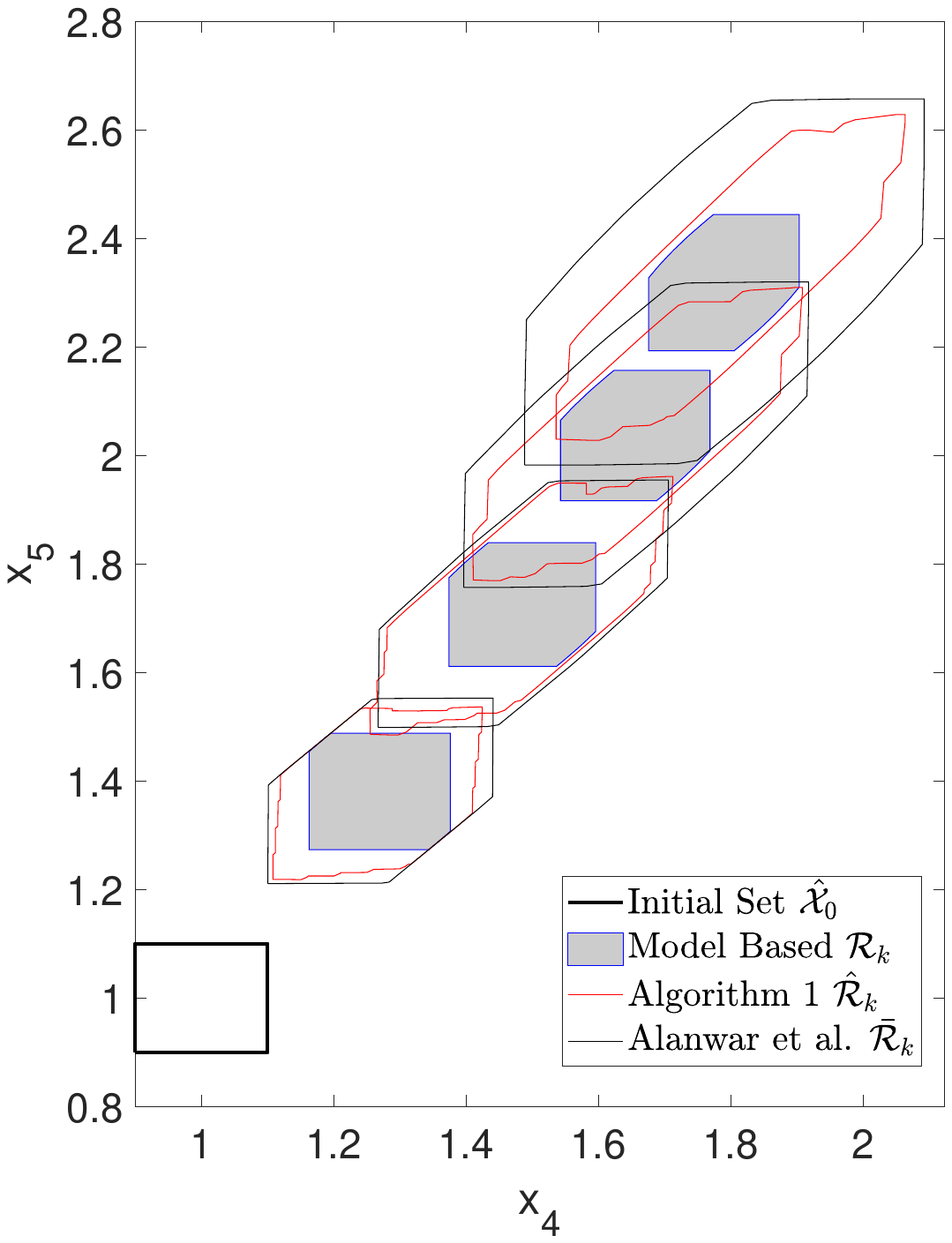}
        \caption{}
        \label{fig:x4x5_b}
    \end{subfigure}
    \caption{Comparative Projection of Reachable Sets from Input-State Data Using Algorithm \ref{alg:offline-online} and Alanwar et al.~\cite{Alanwar2023Datadriven}}
    \label{fig:projSetB}
\end{figure*}

\section{Numerical simulations}\label{sec:numerical-simulations}


Consider a five-dimensional system which is a discretization of the system used in \cite{Alanwar2023Datadriven} with sampling time $0.05$ sec.  All used code to reproduce our results are publicly available\footnotemark. 
The system has the following matrices.
\footnotetext{\href{https://github.com/TUM-CPS-HN/Data-Driven-Nonconvex-Reachability-Analysis}{https://github.com/TUM-CPS-HN/Data-Driven-Nonconvex-Reachability-Analysis}}
\begin{align*}
\Phi_\tr&=\begin{bmatrix}
    0.9323 &  -0.1890   &      0   &      0   &      0 \\
    0.1890 &   0.9323  &       0  &       0   &      0 \\
         0 &        0  &  0.8596  &   0.0430  &        0 \\
         0 &         0   & -0.0430    & 0.8596      &    0 \\
         0 &         0  &        0    &      0   &  0.9048
\end{bmatrix},\\
\Gamma_\tr&=\begin{bmatrix} 
    0.0436&
    0.0533&
    0.0475&
    0.0453&
    0.0476
    \end{bmatrix}^\top.
\end{align*}

In the first experiment, the initial non-convex set is chosen to be $\mathcal{X}_0 = \zono{c_0,G_0,E_0,[~],[~],[~],\id_0}_\mathcal{P}$ , where
\begin{align*}
\begin{split}
c_0=1_{5 \times1},~ G_0=0.1I_5,~ E_0&=\begin{bmatrix}
    2 &  1   &      0   &      0   &      0 \\
    1 &   2  &       0  &       0   &      0 \\
         0 &        0  &  2  &   1  &        0 \\
         0 &         0   & 1    & 2      &    1 \\
         0 &         0  &        0    &      1   &  2
\end{bmatrix}.
\end{split}
\end{align*}
For the second experiment, a convex initial set $\hat{\mathcal{X}}_0=\zono{1_{5 \times1},0.1 I_5}$ is used to compare reachable sets computed via Algorithm~\ref{alg:offline-online} and the method of Alanwar et al.~\cite{Alanwar2023Datadriven}. In both experiments, the input set is defined as $\mathcal{U}_k = \zono{10, 0.25}$, and random noise is sampled from the zonotope $\mathcal{Z}_w = \zono{0, [0.005 , \dots , 0.005]^\top}$. 


Firstly, the reachable sets obtained from a non-convex initial set using both the nominal system model and Algorithm~\ref{alg:offline-online} are illustrated in Fig.~\ref{fig:projSetA}. The reachable set $\tilde{\mathcal{R}}_k$ corresponds to the one computed using the model set $\mathcal{M}_0^\Sigma$, without applying any set refinement. $\hat{\mathcal{R}}_k$ represents the reachable set computed after performing the first online set refinement using new incoming data through Algorithm~\ref{alg:offline-online}. Owing to the exact set multiplication employed in the algorithm, the non-convexity of the initial set is preserved throughout the reachability propagation process. Furthermore, Proposition~\ref{prop:multi} guarantees that the initial reachability computation remains applicable to non-convex initial sets by employing exact multiplication, thus addressing the limitation of the method proposed by Alanwar et al.~\cite{Alanwar2023Datadriven}, which is confined to convex cases.

Secondly, to evaluate the performance of the proposed algorithm, a comparative analysis was conducted under two scenarios. First, trajectory data collected over the interval $[0, T_1]$ with $N_1$ trajectories were used to perform the offline initialization step of Algorithm~\ref{alg:offline-online}. At time $T_1$, new data over $[T_1, T_2]$ containing $N_2$ trajectories became available, triggering the online set refinement step of Algorithm~\ref{alg:offline-online} and yielding the updated reachable set $\hat{\mathcal{R}}_k$. For consistency and fair comparison with the method of Alanwar et al.~\cite{Alanwar2023Datadriven}, we set $N_1 = N_2$. The reachable set $\bar{\mathcal{R}}_k$ computed using the method of Alanwar et al.~\cite{Alanwar2023Datadriven} is based on the combined dataset $D_a$ over the full interval $[0, T_2]$.

As illustrated in Fig.~\ref{fig:projSetB}, the reachable set $\hat{\mathcal{R}}_k$ computed using the Algorithm~\ref{alg:offline-online} is less conservative than $\bar{\mathcal{R}}_k$ obtained using the method of Alanwar et al.~\cite{Alanwar2023Datadriven}. This improvement stems from the incremental refinement of the model set, where newly acquired data are used to further constrain the previously computed set of models. As a result, a more accurate representation of system behavior is achieved in the form of a CMZ, rather than a MZ.
The use of exact set multiplication between CMZ and CPZ in Algorithm~\ref{alg:offline-online} ensures exact set propagation, thereby further reducing conservativeness.
These results highlight the advantages of the proposed incremental approach over direct one-shot data-driven methods, demonstrating its improved ability to capture system dynamics and perform set propagation.

\section{Conclusion} \label{sec:conclusion}

We presented a data-driven framework for reachability analysis using non-convex set representations. 
We introduced a set refinement strategy that incrementally incorporates newly collected input-state data, enabling iterative model set updates. 
This refinement process is supported by an intersection operation for CMZs, allowing system dynamics to be represented more accurately over time. 
Furthermore, we introduced the CPMZ as a new set representation that enables exact multiplication with CPZs. 
Since CMZs are a special case of CPMZs, the proposed method seamlessly applies to the exact multiplication between CMZs and CPZs. 
This enables the exact propagation of non-convex reachable sets by avoiding the over-approximation errors inherent in the non-exact arithmetic of zonotope operations.
The proposed approach thus achieves improved accuracy and reduced conservativeness in safety verification.
Future work will focus on extending the framework to systems with nonlinear dynamics.

\bibliographystyle{IEEEtran}
\bibliography{ref}

\end{document}